\documentclass[12pt]{article}

\usepackage{listings}
\usepackage{verbatim}
\usepackage{setspace}
\usepackage{sectsty}
\subsectionfont{\normalsize}
\makeatletter
\renewcommand\section{\@startsection{section}{1}{\z@}%
                                  {-2.0ex \@plus -1ex \@minus -.2ex}%
                                  {2.0ex \@plus.2ex}%
                                  {\normalfont\normalsize\bfseries}}
\makeatother
\usepackage{fullpage}
\usepackage{url}
\usepackage{amsmath}
\usepackage{mathdots}
\newcommand{\ang}[1]{\langle#1\rangle}

\newcommand{\xvec}[1]{\ifcase 3{#1} {\ang {x_1,x_2,x_3} } \else 
\ifcase 4{#1} {\ang{x_1,x_2,x_3,x_4}} \else {\ang {x_1,\ldots,x_{#1}}}\fi\fi}
\newcommand{\yvec}[1]{\ifcase 3{#1} {\ang {y_1,y_2,y_3} } \else 
\ifcase 4{#1} {\ang{y_1,y_2,y_3,y_4}} \else {\ang {y_1,\ldots,y_{#1}}}\fi\fi}
\newcommand{\zvec}[1]{\ifcase 3{#1} {\ang {z_1,z_2,z_3} } \else 
\ifcase 4{#1} {\ang{z_1,z_2,z_3,z_4}} \else {\ang {z_1,\ldots,z_{#1}}}\fi\fi}
\newcommand{\vecc}[2]{\ifcase 3{#2} {\ang { {#1}_1,{#1}_2,{#1}_3 } } \else
\ifcase 4{#1} {\ang { {#1}_1,{#1}_2,{#1}_3,{#1}_{4} } }
\else {\ang { {#1}_1,\ldots,{#1}_{#2}}}\fi\fi}
\newcommand{\veccd}[3]{\ifcase 3{#2} {\ang { {#1}_{{#3}1},{#1}_{{#3}2},{#1}_{{#3}3} } } \else
\ifcase 4{#1} {\ang { {#1}_{{#3}1},{#1}_{{#3}2},{#1}_{#3}3},{#1}_{{#3}4} }
\else {\ang { {#1}_{{#3}1},\ldots,{#1}_{{#3}{#2}}}}\fi\fi}
\newcommand{\veccz}[2]{\ifcase 3{#2} {\ang { {#1}_0,{#1}_2,{#1}_3 } } \else
\ifcase 4{#1} {\ang { {#1}_0,{#1}_2,{#1}_3,{#1}_{4} } }
\else {\ang { {#1}_0,\ldots,{#1}_{#2}}}\fi\fi}
\newcommand{\xve}[1]{\ifcase 3{#1} {x_1,x_2,x_3} \else 
\ifcase 4{#1} {x_1,x_2,x_3,x_4} \else {x_1,\ldots,x_{#1}}\fi\fi}
\newcommand{\yve}[1]{\ifcase 3{#1} {y_1,y_2,y_3} \else 
\ifcase 4{#1} {y_1,y_2,y_3,y_4} \else {y_1,\ldots,y_{#1}}\fi\fi}
\newcommand{\zve}[1]{\ifcase 3{#1} {z_1,z_2,z_3} \else 
\ifcase 4{#1} {z_1,z_2,z_3,z_4} \else {z_1,\ldots,z_{#1}}\fi\fi}
\newcommand{\ve}[2]{\ifcase 3#2 {{#1}_1,{#1}_2,{#1}_3} \else
\ifcase 4#2 {{#1}_1,{#1}_2,{#1}_3,{#1}_{4}}
\else {{#1}_1,\ldots,{#1}_{#2}}\fi\fi}
\newcommand{\ved}[3]{\ifcase 3#2 {{#1}_{{#3}1},{#1}_{{#3}2},{#1}_{{#3}3}} \else
\ifcase 4#2 {{#1}_{{#3}1},{#1}_{{#3}2},{#1}_{{#3}3},{#1}_{{#3}4}}
\else {{#1}_{{#3}1},\ldots,{#1}_{{#3}{#2}}}\fi\fi}
\newcommand{\fuve}[3]{
\ifcase 3#2
{{#3}({#1}_1),{#3}({#1}_2,{#3}({#1}_3)} \else
\ifcase 4#2
{{#3}({#1}_1),{#3}({#1}_2),{#3}({#1}_3),{#3}({#1}_4)}
\else
{{#3}({#1}_1),\ldots,{#3}({#1}_{#2})}\fi\fi}
\newcommand{\setmathchar}[1]{\ifmmode#1\else$#1$\fi}
\newcommand{\vlist}[2]{%
	\setmathchar{%
		\compound#2\one{#2}\two
		\ifcompound
			({#1}_1,\ldots,{#1}_{#2})
		\else
			\ifcat N#2
				({#1}_1,\ldots,{#1}_{#2})
			\else
				\ifcase#2
					({#1}_0)\or
					({#1}_1)\or
					({#1}_1,{#1}_2)\or 
					({#1}_1,{#1}_2,{#1}_3)\or
					({#1}_1,{#1}_2,{#1}_3,{#1}_4)\else 
					({#1}_1,\ldots,{#1}_{#2})
				\fi
			\fi
		\fi}}

\newif\ifcompound
\def\compound#1\one#2\two{%
	\def\one{#1}
	\def\two{#2}
	\if\one\two
		\compoundfalse
	\else
		\compoundtrue
	\fi}

\newcommand{\xwe}[1]{\ifcase 3{#1} {x_1\wedge x_2\wedge x_3} \else 
\ifcase 4{#1} {x_1\wedge x_2\wedge x_3\wedge x_4} \else {x_1\wedge \cdots \wedge
x_{#1}}\fi\fi}
\newcommand{\we}[2]{\ifcase 3#2 {\ang { {#1}_1\wedge {#1}_2\wedge {#1}_3 } } \else
\ifcase 4{#1} {\ang { {#1}_1\wedge {#1}_2\wedge {#1}_3\wedge {#1}_{4} } }
\else {\ang { {#1}_1\wedge \cdots\wedge {#1}_{#2}}}\fi\fi}

\newcommand{\st}{\mathrel{:}}

\newcommand{\s}[1]{\s_{#1}}

\newcommand{\monus}{\;\raise.5ex\hbox{{${\buildrel
    \ldotp\over{\hbox to 6pt{\hrulefill}}}$}}\;}

\newcounter{savenumi}

\newtheorem{theoremfoo}{Theorem}[section] 
\newenvironment{theorem}{\pagebreak[1]\begin{theoremfoo}}{\end{theoremfoo}}

\newtheorem{lemmafoo}[theoremfoo]{Lemma}
\newenvironment{lemma}{\pagebreak[1]\begin{lemmafoo}}{\end{lemmafoo}}
\newtheorem{conjecturefoo}[theoremfoo]{Conjecture}

\newtheorem{conventionfoo}[theoremfoo]{Convention}

\newtheorem{porismfoo}[theoremfoo]{Porism}

\newtheorem{gamefoo}[theoremfoo]{Game}

\newtheorem{corollaryfoo}[theoremfoo]{Corollary}

\newtheorem{openfoo}[theoremfoo]{Open Problem}

\newtheorem{exercisefoo}{Exercise}

\newcommand{\fig}[1] 
{
 \begin{figure}
 \begin{center}
 \input{#1}
 \end{center}
 \end{figure}
}

\newtheorem{potanafoo}[theoremfoo]{Potential Analogue}

\newtheorem{notefoo}[theoremfoo]{Note}

\newtheorem{notabenefoo}[theoremfoo]{Nota Bene}

\newtheorem{nttn}[theoremfoo]{Notation}

\newtheorem{empttn}[theoremfoo]{Empirical Note}

\newtheorem{examfoo}[theoremfoo]{Example}

\newtheorem{dfntn}[theoremfoo]{Def}
\newenvironment{definition}{\pagebreak[1]\begin{dfntn}\rm}{\end{dfntn}}

\newtheorem{propositionfoo}[theoremfoo]{Proposition}

\newenvironment{proof}
    {\pagebreak[1]{\narrower\noindent {\bf Proof:\quad\nopagebreak}}}{\QED}

\newcommand{\yyskip}{\penalty-50\vskip 5pt plus 3pt minus 2pt}
\newcommand{\blackslug}{\hbox{\hskip 1pt
        \vrule width 4pt height 8pt depth 1.5pt\hskip 1pt}}
\newcommand{\QED}{{\penalty10000\parindent 0pt\penalty10000
        \hskip 8 pt\nolinebreak\blackslug\hfill\lower 8.5pt\null}
        \par\yyskip\pagebreak[1]}

\newcommand{\BBB}{{\penalty10000\parindent 0pt\penalty10000
        \hskip 8 pt\nolinebreak\hbox{\ }\hfill\lower 8.5pt\null}
        \par\yyskip\pagebreak[1]}

\newtheorem{factfoo}[theoremfoo]{Fact}

\newenvironment{block}{\begin{list}{\hbox{}}{\leftmargin 1em
    \itemindent -1em \topsep 0pt \itemsep 0pt \partopsep 0pt}}{\end{list}}

\begin{document}

\centerline{A Sane Proof that $COL_k \le COL_3$}

\centerline{\bf By William Gasarch}

\begin{abstract}
Let $COL_k$ be the set of all graphs that are $k$-colorable.
It is well known that $COL_k$ is NP-complete. It is also
well known, and easy, to show that if $a \le b$ then $COL_a \le COL_b$.
If $3\le a \le b$ then we also have
$COL_b \le SAT \le COL_a$ which is an insane reduction from $COL_b$
to $COL_a$. In this paper we give a sane reduction from $COL_b$ to $COL_a$.
\end{abstract}


\section{Introduction}

Let $A\le B$ mean $A$ is polynomial-time reducible to $B$.

\begin{definition}
Let $k\ge 2$. $COL_k$ is the set of all graphs that are $k$-colorable
\end{definition}

Karp~\cite{karp21} showed that 
$\{ (G,k) \st G\in COL_k \}$ is NP-complete.
Stockmeyer~\cite{3col}
and Lovasz~\cite{lovasz3col} independently showed that $COL_3$ is
NP-complete.

Assume $3\le a<b$.
It is easy to show that, 
$COL_a \le COL_b$ (add $K_{b-a}$ and an edge from every vertex of $K_{b-a}$ to every vertex of $G$.)
What about $COL_b \le COL_a$? By the Cook-Levin Theorem $COL_b \le SAT$ and since $COL_a$ is NP-complete
$SAT\le COL_a$. Hence $COL_b \le COL_a$.
This reduction works but is insane: we transform a graph to a formula and the formula back to a graph.
Is there a sane reduction $COL_b \le COL_a$? There
is and we present it here. For all $k$ we give a sane 
reduction  for $COL_k \le COL_3$.

A proof that does not use formulas is already known.
Let $HCOL_k$ be the set of all hypergraphs that are $k$-colorable.
Lovasz~\cite{lovasz3col} showed $COL_k \le HCOL_2 \le COL_3.$
Our proof does not use hypergraphs or formulas.


\section{The Key Gadget}

The following gadget is often used to prove that $COL_k$ is NP-complete.

\begin{definition}
$GAD(x,y,z)$ is the graph in Figure 1.
(The vertices that don't have labels are never referred to
so we don't need to label them.)
\end{definition}

We leave the proof of the following easy lemma to the reader.

\begin{lemma}
If $GAD(x,y,z)$ is three colored and $x,y$ get the same color,
then $z$ also gets that color.
\end{lemma}

\noindent
$x\quad\qquad y$\newline
$\qquad\mid\quad\qquad\mid$\newline
$\qquad\mid\quad\qquad\mid$\newline
$\qquad\mid\quad\qquad\mid$\newline
$\circ---\circ$\newline
$\backslash\quad\qquad/$\newline
$\hbox{\ \ }\backslash\qquad/$\newline
$\hbox{\ \ \ }\backslash\quad/$\newline
$\hbox{\ \ \ \ }z$\newline

\noindent
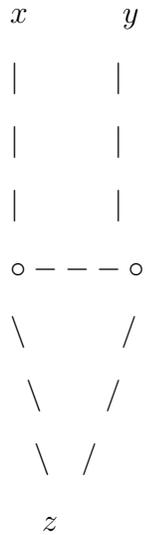
{\bf Figure 1}

\begin{definition}
$GAD(x_1,\ldots,x_k,z)$ consists of $GAD(x_1,x_2,y_1)$,
$GAD(y_1,x_3,y_2)$, $GAD(y_2,x_4,y_3)$, $\ldots$, $GAD(y_{k-3},x_{k-1},y_{k-2})$,
and $GAD(y_{k-2},x_k,z)$.
Aside from  $x_1,\ldots,x_k,z$, the graph
$GAD(x_1,\ldots,x_k,z)$ has $\le 3k$ vertices,
and $\le 5k$ edges.
\end{definition}

We leave the proof of the following easy lemma to the reader.

\begin{lemma}
Let $k\ge 2$.
If $GAD(x_1,x_2,\ldots,x_k,z)$ is three colored and $x_1,\ldots,x_k$ get the same color,
then $z$ also gets that color.
\end{lemma}

\section{The Main Theorem}

\begin{theorem}
Let $k\ge 2$.
$COL_k \le COL_3$ by a simple reduction
that take a graph $G$ with
$n$ vertices and $e$ edges, and produces a graph $G'$ that has 
$\le 2k^2n + 2ke$ vertices and $\le 3k^2n + 2ke$ edges.
\end{theorem}

\begin{proof}
Let $G$ have vertices $v_1,\ldots,v_n$ and edge set $E$.
We construct $G'$:

\begin{enumerate}
\item
There are vertices $T,F,R$ which form a triangle. In any coloring
they have different colors which we call $T,F,R$.
This is 3 vertices and 3 edges.
(We won't count these in the end since our crude upper bounds
on the vertices and edges in $G'$
will clearly be over by at least 3.)
\item
For $1\le i\le n$ and $1\le j\le k$ there is a vertex $v_{ij}$. 
All of these will be connected by an edge to vertex $R$.
This requires be $kn$ vertices and $kn$ edges.
\begin{enumerate}
\item
For all $1\le i\le n$ our intent is:
$v_{ij}$ is colored $T$ means that vertex $v_i$ in $G$ is colored $j$;
$v_{ij}$ is colored $F$ means that vertex $v_i$ in $G$ is not colored $j$.
\item
For all $1\le i\le n$ we need that {\it at least one} of
$v_{i1},\ldots,v_{in}$ is colored $T$.
Hence we need it to not be the case that $v_{i1},v_{i2},\ldots,v_{in}$ are all
colored $F$. We place the gadget $G(v_{i1},\ldots,v_{in},T)$ in the graph.
If $v_{i1},\ldots,v_{in}$ are all colored $F$ then this gadget will not be 3-colorable.
This is
$\le 3kn$  vertices and $\le 5kn$ edges.
\item
For all $1\le i\le n$ we need that {\it at most one} of 
$v_{i1},\ldots,v_{ik}$ is colored $T$.
Hence we need that, for each pair of vertices $v_{ij_1},v_{ij_2}$ at most one is colored $T$.
For each $1\le j_1 < j_2 \le k$ we place the gadget $GAD(v_{ij_1},v_{ij_2},F)$.
This is $n\binom{k}{2}\times 2\le k^2n$ vertices and $n\binom{k}{2}\times 5\le 2.5k^2n$ edges.
\end{enumerate}
\item
For each edge $(v_i,v_j)$ in the original graph we want to make sure that
$v_i$ and $v_j$ are not the same color. 
Place the gadgets 
$GAD(v_{i1},v_{j1},F)$,
$GAD(v_{i2},v_{j2},F)$,$\ldots$,
$GAD(v_{ik},v_{jk},F)$.
This is $2ke$ vertices and $5ke$ edges.
\end{enumerate}

Note that the number of vertices in $G'$ is $\le kn + 3kn + k^2n +2ke \le 2k^2n + 2ke$ vertices
and $\le kn + 5kn + 2.5k^2n + 2ke \le 3k^2n + 2ke$ edges.

Clearly $G$ is $k$-colorable iff $G'$ is 3-colorable.
\end{proof}

\section{Open Problem}

Our reduction takes a graph on $n$ vertices and $e$ edges and produces a graph on
$O(n+e)$ vertices and $O(n+e)$ edges. Can this be improved? For example, is there a reduction
that yields a graph with $O(n+ \sqrt{e} )$ vertices? $O(n)$ vertices? 

\section{Acknowledgment}

I would like to thank my students who protested
that the reduction $COL_b \le SAT \le COL_a$ was insane and demanded a sane reduction.


\begin{thebibliography}{1}

\bibitem{karp21}
R.~Karp.
\newblock Reducibility among combinatorial problems.
\newblock In {\em Complexity of computer computations}, pages 85--103, 1972.

\bibitem{lovasz3col}
L.~Lovasz.
\newblock Coverings and colorings of hypergraphs.
\newblock In {\em Proc. of the 4th Southeastern Conference on Combinatorics,
  Graph Theory, and Computing}, pages 3--12, 1973.
\newblock \url{www.cs.elte.hu/~lovasz/scans/covercolor.pdf}.

\bibitem{3col}
L.~Stockmeyer.
\newblock Planar 3-colorability is polynomial complete.
\newblock {\em SIGACT News}, 5(1), 1973.

\end{thebibliography}

\end{document}